%% file: paperArxiv_1_.tex
\documentclass[runningheads]{llncs}
\usepackage[T1]{fontenc}
\usepackage{graphicx}
\usepackage{amssymb}
\usepackage{amsmath}
\usepackage{hyperref}
\usepackage{cleveref}
\usepackage{dsfont}
\usepackage{xspace}
\usepackage{tikz}
\usepackage{marvosym}
\usepackage{thm-restate}
\usetikzlibrary{intersections,graphs,trees,fit,shapes.geometric,quotes,arrows.meta,backgrounds,matrix} 
\usepackage{csquotes}
\usepackage{comment}
\usepackage{booktabs}
\usepackage{array}
\usepackage{float}
\usepackage{multicol}
\usetikzlibrary{patterns}
\usetikzlibrary{decorations,arrows}
\usetikzlibrary{decorations.pathmorphing}
\usepgflibrary{decorations.pathreplacing} 
\usetikzlibrary{decorations.text}
\usetikzlibrary{intersections,graphs,trees,fit,shapes.geometric,quotes,arrows.meta,backgrounds,matrix} 

\usetikzlibrary{decorations.pathreplacing,calligraphy}
\usetikzlibrary{calc}
\usetikzlibrary{decorations}
\usetikzlibrary{shapes.geometric, arrows.meta}
\usetikzlibrary{graphs}
\usetikzlibrary{matrix}
\usetikzlibrary{decorations.pathreplacing}
\usepackage{xltabular}
\usepackage{bm}
\usepackage{makecell}
\usepackage{algorithm}
\usepackage[noend]{algpseudocode}
\usepackage[normalem]{ulem}
\usepackage{nicefrac}
\usepackage{dsfont}
\newcommand{\bax}{\bar{x}}
\newcommand{\ginfB}{g_\infty(B)}
\begin{document}
\input{preambleIwoca}
\title{Equivalent Instances for Scheduling and Packing Problems}

%
%
\author{Klaus Jansen\inst{1}\orcidID{0000-0001-8358-6796} \and
Kai Kahler\inst{1}\orcidID{0000-0002-8066-4004} \and
Corinna Wambsganz\inst{1}\orcidID{0009-0002-4820-9017}}
\authorrunning{K. Jansen et al.}
%
\institute{Kiel University, Department of Computer Science,  Germany
\email{\{kj,cwa\}@informatik.uni-kiel.de},
\email{kai.kahler@web.de}
}
\maketitle              
\begin{abstract}
Two instances $(I,k)$ and $(I',k')$ of a parameterized problem $P$ are equivalent if they have the same set of solutions (\textit{static} equivalent) or if the set of solutions of $(I,k)$ can be constructed by the set of solutions for $(I',k')$ and some computable pre-solutions. If the algorithm constructing such a (static) equivalent instance whose size is polynomial bounded runs in \textit{fixed-parameter tractable (FPT)} time, we say that there exists a \textit{(static) equivalent instance} for this problem. If the algorithm runs in polynomial time, does not have to guarantee the same set of solutions and the output size of the algorithm is upper bounded by $g(k)$ for some computable function $g$, the algorithm is a \textit{kernel} for this problem. 

In this paper we present (static) equivalent instances for Scheduling and Knapsack problems. We improve the bound for the $\ell_1$-norm of an equivalent vector given by Eisenbrand, Hunkenschröder, Klein, Koutecký, Levin, and Onn and show how this yields equivalent instances for \textit{integer linear programs (ILPs)} and related problems. We obtain an 

$O(MN^2\log(NU))$ static equivalent instance for feasibility ILPs where $M$ is the number of constraints, $N$ is the number of variables and $U$ is an upper bound for the $\ell_\infty$-norm of the smallest feasible solution. With this, we get an $O(n^2\log(n))$ static equivalent instance for Knapsack where $n$ is the number of items. Moreover, we give an $O(M^2N\log(NM\Delta))$ kernel for feasibility ILPs where $\Delta$ is an upper bound for the $\ell_\infty$-norm of the given constraint matrix.

Using balancing results by Knop et al., the \textsc{ConfILP} and a proximity result by Eisenbrand and Weismantel we give an $O(d^2\log(p_{\max}))$ equivalent instance for \textsc{LoadBalancing}, a generalization of scheduling problems. Here $d$ is the number of different processing times and $p_{\max}$ is the largest processing time. 

\end{abstract}
\section{Introduction}
Let $P$ be a parameterized problem and let $(I,k)$ be an instance of $P$. An \textit{static equivalent instance} is an instance $(I',k')$ of $P$ that has exactly the same solutions as $(I,k)$. If we construct an instance $(I',k')$ such that the solutions of $(I,k)$ consist of some computable pre-solutions combined with solutions of $(I',k')$, we denote $(I,k)$ as equivalent instance. In this paper, we allow algorithms that construct (static) equivalent instances to run in \textit{fixed-parameter tractable (FPT)} time. If we require polynomial running time, but only the equivalence $(I,k) \in P$ if and only if $(I',k') \in P$, we get the definition of a kernel. A \textit{kernel} for a parameterized problem $P$ is an algorithm $\mathcal{A}$ that, given an instance $(I,k)$ of $P$, runs in polynomial time and returns an instance $(I',k')$ of $P$ with $(I,k) \in P$ if and only if $(I',k') \in P$. Moreover, the output size of $\mathcal{A}$ is required to be upper bounded by $g(k)$ for some computable function $g$\ \cite{CFKLMPPS15}. The study of kernelization is an important part of parameterized complexity. A decidable problem admits a kernel if and only if it is FPT~\cite{CFKLMPPS15}. For this reason, studies focus on finding a kernel or showing that a problem does not admit a kernel. There have been many new results and new techniques in recent years~\cite{DBLP:journals/eatcs/Kratsch14,DBLP:conf/birthday/LokshtanovMS12}.
This paper proves the existence of small static equivalent instances for \KS and related problems. We improve the bound for the $\ell_1$-norm of an equivalent vector of Eisenbrand, Hunkenschröder, Klein, Koutecký, Levin and Onn~\cite{EHKKLO23} from $N(4N\Delta)^{N-1}$ to $N(2\sqrt{N}\Delta)^{N-1}$ and show how this yields static equivalent instances for ILPs and related problems. For future work, we leave open the question whether the algorithm for the static equivalent instances could be adapted to run in polynomial time and thus, whether this leads to small kernels with the same solution set. Alternatively, we obtain a kernel for feasibility ILPs using a proximity result by Eisenbrand and Weismantel~\cite{EW19}. Moreover, using this proximity result, a balancing result by Knop, Koutecký, Levin, Mnich and Onn~\cite{KKLMO23} and the \confip we give an equivalent instance for \LB on identical machines. After the problem definitions, we give an overview of related work and our results. In \Cref{tab:results} our (static) equivalent instance results are compared to known kernelization results.

\subsection{Problem Definitions}
We denote by $[n] := \{i\ |\ 1 \leq i \leq n\}$. First, we define (static) equivalent instances, the main concept of this paper.
\begin{definition}
    An instance $(I',k')$ of a parameterized problem $P$ is equivalent to an instance $(I,k)$ of this problem if and only if we have for all inputs $x$, for some computable pre-solution $p(x)$ and for some computable function $f$: 
    
    $x=f(p(x),x')$ is a solution of $(I',k')$ if and only if $x'$ is a solution of $(I,k)$.
\end{definition}
If $f$ is the projection on the second argument, thus $f(p(x),x')=x'$ and $x=x'$, we say that the corresponding equivalent instance is a \textit{static} equivalent instance. Note that algorithms which construct (static) equivalent instances of polynomial size are allowed to run in FPT time here, but they preserve the set of valid solutions (exactly). However, if such an algorithm runs in polynomial time, has limited output size and does not necessarily keep a valid solution, this algorithm is a \textit{kernel}:

\begin{definition}[\cite{CFKLMPPS15}]
    A kernelization algorithm, or simply a kernel, for a parameterized problem $P$ is an algorithm $\mathcal{A}$ that, given an instance $(I,k)$ of $P$, runs in polynomial time and returns an instance $(I',k')$ of $P$ with $(I,k) \in P$ if and only if $(I',k') \in P$. Moreover, we require that $\text{size}_{\mathcal{A}} \leq g(k)$ for some computable function $g: \mathbb{N} \rightarrow \mathbb{N}$ where $\text{size}_{\mathcal{A}}$ is the output size of the algorithm $\mathcal{A}$.
\end{definition}
In the following we define ILPs and variants of these.
\begin{definition}
    An ILP is a problem of the following form: 
    $$\max\{c^Tx\ |\ \mathcal{A}x = b, l \leq x \leq u, x \in \mathbb{Z}^N\}$$
    for a matrix $\mathcal{A} \in \mathbb{Z}^{M \times N}$, a right hand side $b \in \mathbb{Z}^M$, a cost vector $c \in \mathbb{Z}^N$ and lower and upper bounds $l,u \in \mathbb{Z}^N$. If $c = \mathbf{0}$, then the problem is a feasibility ILP.
\end{definition}

\begin{definition}
    A \KS instance consists of an item set $I$, weights $w_i$ and profits $p_i$ for $i \in I$ and a capacity $C \in \mathbb{N}$. A solution is a subset $I' \subseteq I$ such that $\sum_{i \in I'}{w_i} \leq C$ and $\sum_{i \in I'}{p_i} \geq T$ for a target value $T$.
\end{definition}

This problem is commonly referred to as \textsc{0-1-Knapsack}. If items have an $M$-dimensional weight vector, the problem is referred to as \MDKS. 
\begin{definition}
    A \MDKS instance consists of an item set $I$, weight vectors $w^{(i)}, i \in I$, profits $p_i, i \in I$, and a capacity vector $b \in \mathbb{N}^M$. A solution is a subset $I' \subseteq I$ such that $\sum_{i \in I'}{w^{(i)}_j} \leq b_j$ for all $j \in [M]$  and $\sum_{i \in I'}{p_i} \geq T$ for a target value $T$.
\end{definition}
In the \UKS problem, items may be taken an unlimited number of times. The \SSS problem is a special case of \KS where all items $i \in I$ have profit $p_i = w_i$ and the capacity must be met exactly. In order to introduce the next problem, we first define some terms in the context of scheduling. In scheduling problems we consider, we are given a number of jobs and identical machines. Each job has a processing time. There are $d \in \mathbb{N}$ different processing times (types). The processing time vector $p \in \mathbb{N}^d_{>0}$ describes the processing times and the job vector $n \in \mathbb{N}^d_{>0}$ describes how many jobs of each type are given. Our goal is to assign each job to one machine such that some objective function gets minimized. The assignment of all jobs to machines (one machine per job) is called a \textit{schedule}. In a specific schedule each machine has a set of jobs assigned. The sum of all the processing times of the jobs assigned to a specific machine is the \textit{load} of this machine.

\begin{definition}
    A \LB on identical machines instance consists of a processing time vector $p \in \mathbb{N}^d_{>0}$, a job vector $n \in \mathbb{N}^d_{>0}$, a number of machines $m \in \mathbb{N}_{>0}$, and thresholds $l,u \in \mathbb{N}$. The task is to decide whether there is a schedule in which each machine has a load that lies in $[l,u]$.
\end{definition}

The solution of a \LB instance is a schedule in which each machine has a load that lies in $[l,u]$, if there exists such a schedule or a schedule in which at least one machine has a load that lies not in $[l,u]$, otherwise. In the context of decision problems, \tf{\P}{}{\cmax} and \tf{\P}{}{\cmin} are special cases of \LB. For \tf{\P}{}{\cmax}, we have a lower bound $l = 0$, and for \tf{\P}{}{\cmin}, we have an upper bound of $u = \infty$. The problem \tf{\P}{}{\envy} describes the objective to minimize the difference between the maximal and minimal load of a schedule. The maximal load of a schedule is the maximal load of a machine in this schedule. Finally, we define equivalent vectors.

\begin{definition}[Equivalent Vectors]
    We say that two vectors $w,\bar{w}\in\R^N$ are equivalent on a set $S\subset\R^N$ if and only if for all $x,y\in S$,
    
    $w^Tx\geq w^Ty\iff \bar{w}^Tx\geq\bar{w}^Ty$. 
\end{definition}

\subsection{Related Work}
Harnik and Naor~\cite{DBLP:journals/siamcomp/HarnikN10} showed that reducing all input numbers of \SSS modulo a random prime of magnitude about $2^{2n}$ yields an equivalent instance with error probability exponentially small in $n$. This is a randomized kernel with size $O(n^2)$. However, the question for a deterministic polynomial kernel for \SSS with respect to the number of items remained open. Thus, Etscheidt \etal~\cite{EKMR17} gave deterministic kernel results for \KS and related problems using the coefficient reduction technique by Frank and Tardos~\cite{FT87}. In particular, they gave a kernel of size $O(n^4)$ for \KS and \SSS. They also showed the following lower bound:
\begin{theorem}[Theorem 14 in~\cite{EKMR17}]\label{thm:kernellowerbound}
\SSS and hence also \KS do not admit a kernel of size $\Oh{n^{2-\eps}}$ for any $\eps>0$, unless the \textit{Exponential Time Hypothesis (ETH)} fails.
\end{theorem}
Kernels for \KS have also been studied by Heeger \etal~\cite{HHMS24}, who showed that one cannot (under the standard complexity assumption NP $\nsubseteq$ coNP/poly) compute a kernel of size $a_{\#}^{O(1)}$ or $v_{\#}^{O(1)}$ in polynomial time, where $a_{\#}$ and $v_{\#}$ are the numbers of different item sizes and values, respectively. They also showed that there is a kernel of size $(a_{\#}+v_{\#})^{O(1)}$ that can be computed in polynomial time. Gurski \etal~\cite{GRR19} used a technique by Frank and Tardos~\cite{FT87} to obtain kernels for \KS parameterized by several different parameters, one of which is the number of items $n$. This result is the same that Etscheidt \etal~\cite{EKMR17} obtained, with a kernel of size $O(n^4)$. Kratsch~\cite{DBLP:conf/esa/Kratsch13} showed that feasibility ILPs admit no polynomial kernel when parameterized by both the number of variables $N$ and the number of constraints $M$, unless NP $\subseteq$ coNP/poly.

Regarding \LB, Knop and Koutecký~\cite{DBLP:conf/esa/KnopK22} gave a kernel for \tf{\P}{}{\cmax} with size poly$(d,\log(\pmax))$\footnote{Their proof goes by reducing \tf{\P}{}{\cmax} to Huge $n$-fold IPs, then kernelizing that to size poly$(d,\log(\pmax))$. After that, they use the fact that there must exist a polytime reduction from Huge $n$-fold IP to \tf{\P}{}{\cmax} because both problems belong to NP. Since this last reduction is polytime, the resulting size is still poly$(d,\log(\pmax))$, even though this is not stated in the original paper. We had personal communication with the authors about that.} that can be computed in time $(p_{\max} +|I|)^{\Oh{1}}$. The proof is quite non-trivial and uses the heavy machinery of Huge $n$-fold IPs. 

The \Cref{tab:results} provides an overview of the above mentioned kernel results compared to our (static) equivalent instance results.

\subsection{Our Contribution}
First, we improve the bound for the $\ell_1$-norm of the equivalent vector from $N(4N\Delta)^{N-1}$~\cite{EHKKLO23} to $N(2\sqrt{N}\Delta)^{N-1}$. Using this bound, we prove the existence of $O(n^2\log(n))$ static equivalent instances for \KS and \SSS. So, we almost close the quadratic gap of~\Cref{thm:kernellowerbound}. However, our static equivalent instances for \KS and \SSS are obtained using the improved bound of Eisenbrand, Hunkenschröder, Klein, Koutecký, Levin and Onn~\cite{EHKKLO23}, meaning that we know that there exists an equivalent small instance with the same solution set and that we can compute it in exponential time $(n\Delta)^{\Oh{n}}$, where $\Delta$ is the largest item size or value. In contrast, the algorithm by Frank and Tardos~\cite{FT87} (and hence also the procedure by Etscheidt \etal~\cite{EKMR17}) takes only polynomial time. Note that our static equivalent instance for \KS and the kernel by Etscheidt \etal~\cite{EKMR17} do not contradict the result of Heeger \etal~\cite{HHMS24}, as the total number $n$ of items may be much larger than the number of different item sizes or values.

For feasibility ILPs we present a static equivalent instance of size $\Oh{MN^2\log(NU)}$, where $U$ is an upper bound for the $\ell_\infty$-norm of the smallest feasible solution. This yields the $O(n^2\log(n))$ static equivalent instances for \KS and \SSS, an $\Oh{n^2\log^2(C)\log(n\log(C))}$ equivalent instance for \UKS and an $\Oh{Mn^2\log(n)}$ static equivalent instance for \MDKS. The implications for special forms of ILPs ($n$-fold, $2$-stage) are listed in \Cref{sec:specialilp}. 

Moreover, we get for feasibility ILPs a kernel of size $O(M^2N\log(NM\Delta))$ where $\Delta$ is an upper bound for the $\ell_\infty$-norm of the given constraint matrix.

Our equivalent instance for \LB on identical machines only has size $O(d^2\log(\pmax))$ and, in contrast to Knop and Koutecký~\cite{DBLP:conf/esa/KnopK22}, can be computed in a direct way in time $\pmax^{O(d)}$, while the proof mainly relies on a balancing result by Knop, Koutecký, Levin, Mnich and Onn~\cite{KKLMO23}.

\begin{table}[h]
\begin{tabular}{l|c|c}
Problem & our (\textcolor{blue}{static}) equivalent instance & known kernel bound \\
\hline
\textcolor{blue}{feasibility ILPs} & $\Oh{MN^2\log(NU)}$  & $O(M(N^4 + N^3\log(NU))$\cite{EKMR17} \\
\textcolor{blue}{\KS} & $O(n^2\log(n))$ &  $\tilde{O}((w_\#\cdot p_\#)^5)$ \cite{HHMS24}, $O(n^4)$ \cite{EKMR17}, \cite{GRR19}\\
\textcolor{blue}{\SSS} & $O(n^2\log(n))$ & $O(n^4)$ \cite{EKMR17}, $O(k^4\log(k))$\cite{EKMR17} \\
\textcolor{blue}{MD-\KS} & $\Oh{Mn^2\log(n)}$  & $O(n^5)$\cite{GRR19} \\
LB, \tf{\P}{}{\{\cmax,\cmin\}} & $O(d^2\log(\pmax))$ & poly$(d,\pmax)$~\cite{DBLP:conf/esa/KnopK22}\\
\end{tabular}
\vspace{0.25cm}
\caption{Comparing our (\textcolor{blue}{static}) equivalent instance sizes with already known kernel sizes. MD-\KS denotes \MDKS and LB denotes \LB. The parameter $w_\#$ is the number of different weights and $p_\#$ is the number of different profits. The value $U$ is an upper bound for the $\ell_{\infty}$-norm of the smallest feasible solution, $M$ is the number of constraints and $k$ is the number of different item weights.}
\label{tab:results}
\end{table}

After introducing some preliminaries in \Cref{sec:preliminaries}, we present the static equivalent instances with coefficient reduction in \Cref{sec:eqcoef}. Afterwards, we show in \Cref{sec:pcmaxkernel} the equivalent instance for \LB. 

Due to space constraints, some proofs were moved into the appendix. The corresponding statements are marked with (\Rightscissors).

\section{Preliminaries}\label{sec:preliminaries}
In order to compute (static) equivalent instances, we need some theorems and some more definitions presented in this section. The following \Cref{thm:franktardos} and \Cref{cor:franktardos} describe properties of equivalent vectors.

\begin{theorem}[Frank \& Tardos~\cite{FT87}]\label{thm:franktardos}
    For every $w\in\R^N$ and $\Delta\in\N$, there exists a $\bar{w}\in\Z^N$ such that $\norm{\bar{w}}_{\infty}\leq (N\Delta)^{\Oh{N^3}}$ and $\textup{sign}(w^Tx)=\textup{sign}(\bar{w}^Tx)$ for every $x\in\Z^N$ with $\norm{x}_1\leq \Delta-1$. Moreover, $\bar{w}$ can be computed in time $N^{\Oh{1}}$.
\end{theorem}
\begin{restatable}[Jansen \etal~\cite{JKZ24} or Etscheidt \etal~\cite{EKMR17}]{corollary}{franktardos}\label{cor:franktardos}
    For every $w\in\R^N$, $b\in\R$, $\Delta\in\N$, one can compute $\bar{w}\in\Z^N$, $\bar{b}\in\Z$ with $\norm{\bar{w}}_{\infty},|\bar{b}|\leq (N\Delta)^{\Oh{N^3}}$ in time $N^{\Oh{1}}$ such that for every $x\in[-\Delta,\Delta]^N$, $w^Tx\leq b \iff \bar{w}^Tx\leq \bar{b}$.
\end{restatable}

The following \Cref{lem:blemma} describes a balancing result obtained via techniques used in Knop, Koutecký, Levin, Mnich and Onn~\cite{KKLMO23}.
\begin{restatable}[\Rightscissors]{lemma}{balance}\label{lem:blemma}
    For \tf{\P}{}{\{\cmax,\cmin,\envy\}}, there exists a kernel where the number of jobs of a specific type on a specific machine is bounded by $4d(4\pmax +1)^2$. So the load of every machine is bounded by $4d^2\pmax(4\pmax +1)^2$. The kernelization runs in $\Oh{d}$ time. 
\end{restatable}

We use the following definitions for a finitely generated cone, an integer cone, polyhedral cones, a polyhedron and a polytope.

A set $C$ is called a \emph{finitely generated cone} if there is a finite set $S$ such that $C=\textup{cone}(S):=\left\{y\,\left|\,\exists\lambda\in\R_{\geq0}^S:\,y=\sum_{s\in S}\lambda_s s\right.\right\}$. I.e., $C$ is the set of points that can be written as a conic combination of the points in $S$. Similarly, the \emph{integer cone} of a finite set $S$ is defined as $\textup{int.cone}(S):=\left\{y\,\left|\,\exists\lambda\in\N^S:\,y=\sum_{s\in S}\lambda_s s\right.\right\}$, the only difference being that the scalars have to be integral. \emph{Polyhedral cones} can be written as $\left\{\left.x\in\R^N\,\right|\,Ax\leq0\right\}$ for some matrix $A$. A famous result by Farkas, Minkowski and Weyl shows that a convex cone is polyhedral if and only if it is finitely generated (see Corollary 7.1a in~\cite{Schrijver86}).

A (rational) \emph{polyhedron} can be viewed as the set of solutions to a system of equalities: $P=\left\{\left.x\in\R_{\geq0}^N\,\right|\,Ax=b\right\}$; a \emph{polytope} is a bounded polyhedron. The entries in $A$ and $b$ are always assumed to be integer throughout this work.

\section{Equivalent Instances via Coefficient Reduction}\label{sec:eqcoef}
In this section, we improve a result by Eisenbrand \etal~\cite{EHKKLO23} that is quite similar to the one by Frank and Tardos~\cite{FT87}, but only shows the existence of an equivalent vector with smaller coefficients, seemingly being harder to construct. The proof is quite similar, but we use Cramer's rule instead of bounds for Graver Basis elements. Afterwards, we investigate its implications for (static) equivalent instances.

\subsection{The Theorem by Eisenbrand \etal}

Every polyhedral cone is generated by solutions of certain systems of linear equations:
\begin{proposition}[Proposition 2 in~\cite{EHKKLO23}]\label{prop:minkowski}
Let $C=\left\{\left.x\in\R^N\,\right|\,Ax\geq0\right\}$ be a polyhedral cone and let $S'$ be the set of all solutions to any of the systems $By=b'$, where $B$ consists of $N$ linearly independent rows of the matrix $\begin{pmatrix}A \\ I\end{pmatrix}$ and $b'=\pm e_j$, $j\in[N]$. Then there exists an $S\subseteq S'$ such that $C=\textup{cone}(S)$.
\end{proposition}
Here, $e_k$ is the unit vector with a 1 at position $k$. We use the following result to slightly improve the theorem by Eisenbrand \etal~\cite{EHKKLO23}:

\begin{lemma}[\Rightscissors]\label{lem:generatorsalternative}
Let $C=\left\{\left.x\in\R^N\,\right|\,Ax\geq0\right\}$ be a polyhedral cone, where $A\in\Z^{M\times N}$. Then there is a finite set $S$ of integral vectors such that $C=\textup{cone}(S)$ and $\norm{v}_1\leq N(\sqrt{N}\norm{A}_{\infty})^{N-1}$ for every $v\in S$.
\end{lemma}

We now slightly improve the result by Eisenbrand \etal~\cite{EHKKLO23}, who find a vector $\bar{w}$ with $\norm{\bar{w}}_1\leq N(4N\Delta)^{N-1}$. This improvement comes from applying \Cref{lem:generatorsalternative} instead of Lemma 3 in \cite{EHKKLO23}, which uses a bound for the Graver basis elements. In a few words, we define a cone $C$ such that all elements of the cone are equivalent to the vector $w$. Then, we conclude that there is an equivalent vector $\tilde{w}$ in $C$ which is small enough since the set describing $C$ only consists of small vectors.
\begin{theorem}[Improved version of Theorem 1 in~\cite{EHKKLO23}]\label{thm:reducenonconstructive}
For every $w\in\R^N$, there exists a $\bar{w}\in\Z^N$ such that $\norm{\bar{w}}_1\leq N^2(2\sqrt{N}\Delta)^{N-1}$ and $w$ and $\bar{w}$ are equivalent on $[-\Delta,\Delta]^N$.
\end{theorem}
\begin{proof}
In the following, we define a cone $C$, show that the points in the relative interior of $C$ are equivalent to $w$ on $[-\Delta,\Delta]^N$ and then give a point in $C$ with small enough $\ell_1$-norm.

For any (integer) points $u,v\in[-\Delta,\Delta]^N$ such that $wu\geq wv$, we define the half-space $H(u,v):=\left\{\left.x\in\R^N\,\right|\,(u-v)x\geq0\right\}$; note that the condition is equivalent to $xu\geq xv$, so $H(u,v)$ can be seen as the set of linear objective functions for which $u$ is \enquote{better} than $v$. Clearly, $w$ is then also inside the half-space. Based on this, we define the cone
\[C:=\bigcap_{\substack{(u,v)\in[-\Delta,\Delta]^N\times[-\Delta,\Delta]^N\\wu\geq wv}}H(u,v).\]
As $w$ is in each of the half-spaces, it is also in their intersection, so $w\in C$. Consider some $z$ in the \textit{relative interior}\footnote{The relative interior of $C$ is the set of points in $C$ that fulfill only those constraints of $C$ with equality that have to be fulfilled with equality by all points in $C$. Other constraints are fulfilled with strict inequality.} of $C$; we show that $z$ is equivalent to $w$ on the interval $[-\Delta,\Delta]$. Intuitively, $z$ being in the relative interior of $C$ means that $z\in C$, but $z$ lies only on the border of $C$ if it has to, i.e., if $C$ is flat in that direction. So $z$ only fulfills a constraint $(u-v)z\geq0$ of $C$ with equality if both $(u-v)z\geq0$ and $(v-u)z\geq0$ are present in the description of $C$. 

Let $x,y\in[-\Delta,\Delta]^N$ and first assume that $wx\geq wy$. Then the half-space $H(x,y)$ is part of $C$ and because $z\in C$, we have $(x-y)z\geq0$, which is equivalent to $zx\geq zy$. For the other direction, indirectly assume that $wx<wy$. Then $H(y,x)$ is part of $C$'s description, but $H(x,y)$ is not. If $zx>zy$, that contradicts $z$ being in the half-space $H(y,x)$, so assume $zx=zy$. As we chose $z$ to be in the relative interior, but $(x-y)z=(y-x)z=0$, both half-spaces $H(x,y)$ and $H(y,x)$ have to be part of $C$'s description, which is a contradiction. So we have $zx<zy$, showing the equivalence of $z$ and $w$ on $[-\Delta,\Delta]^N$.

Now, using \Cref{lem:generatorsalternative}, there is a description of $C$, where each generator has $\ell_1$-norm at most $N(\sqrt{N}2\Delta)^{N-1}$. This follows, because the coefficients in $C$'s polyhedral description are bounded by $2\Delta$ in $\ell_\infty$-norm. Let $S$ be the set of generators and $S'\subseteq S$ a maximal subset of generators that are linearly independent. Note that $|S'|\leq N$. We now give a point inside $C$'s relative interior, namely $\bar{w}:=\sum_{v\in S'}v$. So by the cardinality bound for $S'$ and the bound for the generators, $\norm{\bar{w}}_1\leq N^2(\sqrt{N}2\Delta)^{N-1}$. It is only left to show that $\bar{w}$ lies in the relative interior of $C$. Since $\bar{w}$ is a sum of generators of $C$, it clearly lies in $C$. Now, assume that $x\bar{w}=y\bar{w}$ for two vectors $x,y\in[-\Delta,\Delta]^N$, i.e., $\bar{w}$ lies on a border of $C$. At least one of $H(x,y)$ and $H(y,x)$ must be part of $C$'s description; w.l.o.g. assume that $H(x,y)$ is. Hence, $(x-y)v\geq0$ holds for all $v\in C$, in particular also for all $v\in S'$, while $(x-y)\bar{w}=0$. But since $\bar{w}=\sum_{v\in S'}v$ by definition, $(x-y)v=0$ has to hold for all the $v\in S'$; otherwise, zero could not be achieved by their sum, as they cannot cancel each other out because of $(x-y)v\geq0$. So now we have $(x-y)v=0$ for all $v\in S'$ and since the vectors in $S'$ generate $C$, each vector in $C$ can be written as a conic combination (i.e., a linear combination with non-negative coefficients) of the vectors in $S'$. Hence, every vector $v\in C$ also fulfills $(x-y)v=0$.

We have shown that if a constraint is tight for $\bar{w}$, it is tight for all points in $C$ and that $\bar{w}$ lies in $C$. So $\bar{w}$ lies in the relative interior of $C$ and is hence equivalent to $w$ on the interval $[-\Delta,\Delta]^N$ by the above observation. It also fulfills the bound $\norm{\bar{w}}_1\leq N^2(\sqrt{N}2\Delta)^{N-1}$, as argued above. Hence, we have found a suitable vector $\bar{w}$, which concludes the proof.
\end{proof}

Eisenbrand~\cite{EHKKLO23} also give a corresponding lower bound for the $\ell_1$-norm of an equivalent vector $\bar{w}$, namely $(N\Delta)^{N-1}$. However, there appears to be a slight mistake in the proof: The interval $\mathcal{D}=[-\Delta,\Delta]^N$ they define does not contain the points $v^{(i)}$ they define, as they have entries as large as $N\Delta$.\footnote{Note that they use $N$ for the value bound (instead of $\Delta$) and $n$ for the dimension (instead of $N$).} The proof works, however, if one removes the $N$ from the definition of the points $v^{(i)}$, as we show in the following. Note that the (false) lower bound (Theorem 2 in \cite{EHKKLO23}) would contradict our improved upper bound $N^2(\sqrt{N}2\Delta)^{N-1}$ from \Cref{thm:reducenonconstructive}. This can be seen by setting $\Delta=1$ and $N$ large enough.
\begin{theorem}[Corrected version of Theorem 2 in~\cite{EHKKLO23}, \Rightscissors]\label{thm:reducibilitylowerbound}
Let $N,\Delta\in\Z_{>0}$ and $w=(\Delta^0,\Delta^1,\Delta^2,\hdots,\Delta^{N-1})$. There is no vector $\bar{w}\in\Z^N$ with $\norm{\bar{w}}_1<\Delta^{N-1}$ such that $w$ and $\bar{w}$ are equivalent on the interval $[-\Delta,\Delta]^N$.
\end{theorem}

Recall that the original result by Frank and Tardos (\Cref{thm:franktardos}) is about the sign of the two vectors at every position. The concept of equivalence implies this property:

\begin{lemma}[\Rightscissors]\label{lem:sign}
Let $w,\bar{w}\in\Z^N$. If $w$ and $\bar{w}$ are equivalent on some interval $[-\Delta,\Delta]$ with $\Delta\geq1$, then the $\textup{sign}(w_i)=\textup{sign}(\bar{w}_i)$ for every $i\in[N]$.
\end{lemma}
This also means that zero-entries become zero-entries. Note that the vector from the result by Frank and Tardos~\cite{FT87} can be computed in polynomial time, while the proof of \Cref{thm:reducenonconstructive} is non-constructive. We now show that it can indeed be computed, though in exponential time.

\begin{theorem}[\Rightscissors]\label{thm:reduceconstructive}
The vector $\bar{w}$ from \Cref{thm:reducenonconstructive} can be computed in time $(N\Delta)^{\Oh{N}}$.
\end{theorem}

\subsection{Equivalent Instances for ILPs and Related Problems}
Using the following proximity result by Eisenbrand and Weismantel~\cite{EW19}, we know a fixed part of the solution and thus, can reduce our problem by this fixed part which leads us to a kernel for feasibility ILPs.
\begin{theorem}[Eisenbrand and Weismantel~\cite{EW19}]\label{thm:eisweis}
    Consider an ILP of the form $\max\{\left.c^Tx\,\right|\,\mathcal{A}x=b,\,x\in\N^N\}$ with $M$ constraints and $\Delta=\norm{\mathcal{A}}_\infty$. Then for every optimal (vertex) solution $x^*$ of the LP-relaxation (where $x\in\R_{\geq0}^N$), there exists an optimal solution $z^*$ of the ILP such that $\norm{x^*-z^*}_1\leq M(2M\Delta+1)^M$.
\end{theorem}

\begin{theorem}[\Rightscissors]\label{thm:ilpkernel2}
Feasibility ILPs have a kernel of size $\Oh{M^2N\log(NM\Delta)}$, where $\Delta$ is an upper bound for the $\ell_\infty$-norm of the given constraint matrix.
\end{theorem} 

We now investigate the implications of \Cref{thm:reducenonconstructive} for (static) equivalent instances. The core idea is to encode the problem as a vector and then use \Cref{thm:reducenonconstructive} to get an equivalent small vector. We show that this small vector gives indeed a small static equivalent instance of our problem.
\begin{theorem}\label{thm:ilpkernel}
Feasibility ILPs have a static equivalent instance of size $\Oh{MN^2\log(NU)}$, where $U$ is an upper bound for the $\ell_\infty$-norm of the smallest feasible solution.
\end{theorem}
\begin{proof}
Consider a feasibility ILP of the form
\begin{align*}
    \mathcal{A}x&= b \\
    x&\in\N^N
\end{align*}
with $\mathcal{A}\in\Z^{M\times N}$ and let $U$ be an upper bound for the $\ell_\infty$-norm of the smallest feasible solution to the ILP. Set $\Delta:=U$, consider any constraint $\mathcal{A}_i x= b_i$ (where $\mathcal{A}_i$ is the $i$-th row of $\mathcal{A}$) and define $w_i:=\begin{pmatrix}
\mathcal{A}_i^T \\ b_i
\end{pmatrix}$. Apply the reduction from~\Cref{thm:reducenonconstructive} to $w_i$ and $\Delta$ to obtain a vector $\bar{w}_i$ (consisting of $\bar{\mathcal{A}}_i^T$ and $\bar{b}_i$) with $\norm{\bar{w}_i}_1\leq k^2(2\sqrt{k}\Delta)^{k-1}$, where $k:=N+1$ is the dimension of $w_i$.

We now show that after doing this transformation for every constraint, $x\in\N^N$ is a solution of the original ILP if and only if it is a solution of the ILP described by the vectors $\bar{w}_i$: Consider any constraint $i$ and a solution $x\in\N^N$ satisfying $\norm{x}_\infty\leq U$. We have:
\begin{align*}
    && \mathcal{A}_ix&= b_i \\
    \iff&& \begin{pmatrix}\mathcal{A}_i^T \\ b_i\end{pmatrix}\begin{pmatrix} x \\ 0\end{pmatrix}&= \begin{pmatrix}\mathcal{A}_i^T \\ b_i\end{pmatrix}\begin{pmatrix}\mathbf{0} \\ 1\end{pmatrix} \\
    \iff&& w_i\begin{pmatrix} x \\ 0\end{pmatrix}&= w_i\begin{pmatrix}\mathbf{0} \\ 1\end{pmatrix} \\
    \iff&& \bar{w}_i\begin{pmatrix} x \\ 0\end{pmatrix}&= \bar{w}_i\begin{pmatrix}\mathbf{0} \\ 1\end{pmatrix} \\
    \iff&& \begin{pmatrix}\bar{\mathcal{A}}_i^T \\ \bar{b}_i\end{pmatrix}\begin{pmatrix} x \\ 0\end{pmatrix}&= \begin{pmatrix}\bar{\mathcal{A}}_i^T \\ \bar{b}_i\end{pmatrix}\begin{pmatrix}\mathbf{0} \\ 1\end{pmatrix} \\
    \iff&& \bar{\mathcal{A}}_ix&=\bar{b}_i
\end{align*}
The third equivalence follows from the equivalence of $w_i$ and $\bar{w}_i$ for vectors of $\ell_\infty$-norm at most $\Delta=U$. Since this holds for every constraint, every solution of the ILP defined by the $w_i$ is a solution of the ILP defined by the $\bar{w}_i$ and vice versa. The vectors $\bar{w}_i$ combined to a single vector $\bar{w}$ encode the whole ILP and are hence a static equivalent instance of size
\begin{align*}
O(Mk\log(\norm{\bar{w}}_\infty))&\leq O(Mk\log(\norm{\bar{w}_i}_1))\\
&\leq O(Mk\log(k^2(2\sqrt{k}\Delta)^{k-1}))\\
&=O(Mk^2\log(k\Delta)) \\
&=O(MN^2\log(NU))
\end{align*}
which concludes the proof.
\end{proof}

The upper bound $U$ for the $\ell_{\infty}$-norm of the smallest feasible solution can be upper bounded by $N(M\norm{\mathcal{A}}_{\infty})^{2M+3}(1 + \norm{b}_{\infty})$ by Theorem 13.4 in~\cite{PS82}. 

A static equivalent instance for optimality ILP can be obtained in a similar manner: If we need to check some target function, we can add this constraint to the ILP and add corresponding slack variables.
This yields static equivalent instances for several other problems that can be formulated as an ILP:

\begin{corollary}[\Rightscissors]\label{cor:knapsackkernel}
\KS has a static equivalent instance of size $\Oh{n^2\log(n))}$.
\end{corollary}

\begin{corollary}[\Rightscissors]\label{cor:ssskernel}
\SSS has a static equivalent instance of size $\Oh{n^2\log(n)}$.
\end{corollary}

\begin{corollary}[\Rightscissors]\label{cor:ukskernel}
\UKS has an equivalent instance of size $\Oh{n^2\log^2(C)\log(n\log(C))}.$
\end{corollary}

Note that this equivalent instance is only useful if $T$, the threshold for the value of the solution, is much larger than $C$, the capacity: The equivalent instance only gets rid of the encoding of $T$, not of $C$. 

\begin{corollary}[\Rightscissors]\label{cor:mdimkskernel}
\MDKS has a static equivalent instance of size $\Oh{Mn^2\log(n)}$.
\end{corollary}

\section{An Equivalent Instance for \texorpdfstring{\LB}{LoadBalancing}}\label{sec:pcmaxkernel}
We now show an equivalent instance for \LB on identical machines that uses a balancing result by Knop, Koutecký, Levin, Mnich and Onn~\cite{KKLMO23}, the \confip and \Cref{thm:eisweis}. With the balancing result, we can pre-schedule many jobs. Afterwards, we define the ILP with all possible configurations and use the proximity result in order to get a fixed part of the solution. The equivalent instance describes the remaining open part of the solution.

\begin{theorem}\label{thm:lbkernel}
    \LB on identical machines has an equivalent instance of size $O(d^2\log(\pmax))$ that can be computed in time $\pmax^{O(d)}$.
\end{theorem}
\begin{proof}
    Given an instance $I$ of \LB on identical machines with thresholds $\ell$ and $u$, we first use the preprocessing algorithm by Knop, Koutecký, Levin, Mnich and Onn~\cite{KKLMO23} (\Cref{lem:blemma}) to pre-schedule many of the jobs optimally in time $O(d)$. For the resulting instance $I'$, we are guaranteed that in some optimal schedule, on each machine, there are at most $4d(4\pmax +1)^2$ jobs of each job type scheduled. It has $m'=m$ machines, $n'_i\leq 4md^2(4\pmax +1)^2$ jobs of type $i\in[d]$ and thresholds $\ell',u'\leq 4d^2 \pmax (4\pmax +1)^2$.
    Now, consider the corresponding \confip
    \begin{align*}
        \matrix{c_1 & \hdots & c_{|\mathcal{C}|} \\
        1 & \hdots & 1}
        x &= \matrix{n'\\m'}\\
        x_c&\in\N\quad\forall c\in\mathcal{C}
    \end{align*}
    where
    $$\mathcal{C}:=\left\{\left.c\in\N^d\,\right|\,\ell'\leq p^Tc\leq u',\,\norm{c}_\infty\leq 4d(4\pmax +1)^2 \right\}$$
    is the set of configurations we need to consider for the reduced instance $I'$. We compute a vertex solution $x^*$ of the LP-relaxation of the \confip using, e.g., the algorithm by Tardos~\cite{Tardos86}. The (I)LP has $M=d+1$ constraints, $N=|\mathcal{C}|\leq (4d(4\pmax +1)^2)^d$ variables, the largest coefficient in the matrix is bounded by $\Delta= 4d(4\pmax +1)^2$ and the largest coefficient in the right-hand-side is bounded by $\norm{b}_\infty=\max\{\nmax,m\}$.
    Hence, the algorithm by Tardos~\cite{Tardos86} takes time 
    $$N^{O(1)}\log(\Delta)\leq \pmax^{O(d)}.$$
    By the proximity result of Eisenbrand and Weismantel (\Cref{thm:eisweis}), there exists an optimal solution $z^*$ of the \confip such that 
    $$\norm{x^*-z^*}_1\leq M(2M\Delta+1)^M\leq (d+1)(2(d+1)4d(4\pmax +1)^2+1)^{d+1}.$$
    It is well-known that vertex solutions have at most $M$ non-zero components, meaning that $x^*$ has at most $d+1$ non-zero components (cf. \cite{Schrijver86}). As $x^*$ is a solution of the LP-relaxation, it fulfills the constraint $\sum_{c\in\mathcal{C}}x^*_c=m$. Due to the proximity bound, there exists an optimal solution $z^*$ that only deviates from $x^*$ by $(d+1)(2(d+1)4d(4\pmax +1)^2+1)^{d+1}$ in every component. Subtracting this term from each component still leaves $m-(d+1)(d+1)(2(d+1)4d(4\pmax +1)^2+1)^{d+1}$ machines on which $z^*$ assigns the same configuration as $x^*$, and we know which configurations are assigned: Configuration $c\in\mathcal{C}$ is assigned precisely $\max\{0,\lfloor x^*_c\rfloor-(d+1)(2(d+1)4d(4\pmax +1)^2+1)^{d+1}\}$-times. Note that this approach is quite similar to the one used in the proximity parts of \cite{JKZ24,JK24}. 

    This leaves us with at most $(d+1)(d+1)(2(d+1)4d(4\pmax +1)^2+1)^{d+1}$ unassigned machines. The remaining instance, say $I''$, hence has
    \begin{itemize}
        \item $m''=(d+1)(d+1)(2(d+1)4d(4\pmax +1)^2+1)^{d+1}$ machines, 
        \item $n''_i\leq m''4d^2(4\pmax +1)^2$ jobs of type $i\in[d]$ and
        \item thresholds $\ell'',u''\leq 4d^2\pmax(4\pmax +1)^2$.
    \end{itemize}
    So the encoding length of $I''$, and hence the size of the equivalent instance, is bounded by:
    \begin{align*}
        &\log(m'')+\log(l'')+\log(u'')+d\log(p_{\max}+n''_{\max})\\
        =&\log((d+1)(d+1)(2(d+1)4d(4\pmax +1)^2+1)^{d+1})+2\log(4d^2\pmax(4\pmax +1)^2)\\
        &+d\log(p_{\max}+(d+1)(d+1)(2(d+1)4d(4\pmax +1)^2+1)^{d+1}4d^2(4\pmax +1)^2)\\
        \leq&O(d^2\log(\pmax))
    \end{align*}
    Again, recall that $d\leq\pmax$. In a pre-solution we can save the assignment of jobs to machines done in the preprocessing algorithm as well as the known assignment resulting from the vertex solution of the LP-relaxation of the \confip. Such a pre-solution combined with the solution of the equivalent instance $I''$ results in a solution for the original \LB instance. Computing this equivalent instance takes time $O(d) + \pmax^{O(d)} + O(d|\mathcal{C}|)\leq \pmax^{O(d)}$ for balancing, the LP-solver and pre-assigning configurations, concluding the proof.
\end{proof}

In the context of decision problems, \tf{\P}{}{\cmax} and \tf{\P}{}{\cmin} are special cases of \LB. \tf{\P}{}{\envy} is equivalent to \LB if we assume that $\ell$ and $u$ are given. Hence, \Cref{thm:lbkernel} also yields the following result:
\begin{corollary}
    The problems \tf{\P}{}{\{\cmax,\cmin\}} have an equivalent instance of size $O(d^2\log(\pmax))$ that can be computed in time $\pmax^{O(d)}$. The same holds for \tf{\P}{}{\envy} if $\ell$ and $u$ are given.
\end{corollary}

\section{Conclusion}\label{sec:conclusion}
We slightly improved the bound for the $\ell_1$-norm of the equivalent vector from $N(4N\Delta)^{N-1}$ to $N(2\sqrt{N}\Delta)^{N-1}$ and showed how this yields (static) equivalent instances for ILPs and related problems. Alternatively, we gave a kernel for feasibility ILPs. Moreover, we presented an equivalent instance for \LB and obtained with this an equivalent instance for \tf{\P}{}{\{\cmax,\cmin\}}.

\begin{credits}
\subsubsection{\ackname} This study was funded by the Deutsche Forschungsgemeinschaft (DFG, German Research Foundation) – Project Number 453769249. We thank Martin Koutecký for his very helpful comments and guidance regarding balancing and the kernel of $\PCmax$.
\end{credits}


\bibliography{bibIwoca}

\appendix
\section{Equivalent Instances via Coefficient Reduction}
In \Cref{sec:proofs} we list the missing proofs of \Cref{sec:eqcoef} and in \Cref{sec:specialilp} we investigate the implications of \Cref{thm:ilpkernel} for special ILPs, namely \twostage and \nfold.

\subsection{Proofs}\label{sec:proofs}
\begin{proof}[\Cref{lem:generatorsalternative}]
Let $S$ be the set of generators given by \Cref{prop:minkowski}. We scale each of the generators in $S$ to an integer vector.

Consider some $v\in S$ and let $B$ be the corresponding sub-matrix of $\begin{pmatrix}A \\ I\end{pmatrix}$ such that $v$ is the unique solution to $Bx=\pm e_k$ for some $k\in[N]$. By Cramer's rule~\cite{DBLP:books/lib/ChabertB99,Cramer1750}, the $i$-th entry of $v$ is given by:
\begin{align*}
    v_i=\frac{\det{B^{(i)}}}{\det{B}}
\end{align*}
Here, $B^{(i)}$ is the matrix $B$ but with the $i$-th column replaced by the right-hand-side vector $\pm e_k$. It is important to note that the denominator is the same for all entries in $v$, so scaling the whole vector by $\det{B}$ yields an integer vector $z$ with entries bounded by 
\begin{align*}
    z_i\ &= \det{B^{(i)}}\leq \prod_{j=1}^N\norm{B_j^{(i)}}_2 =\left(\prod_{\substack{j=1 \\ j\neq i}}^N\norm{B_j^{(i)}}_2\right)\norm{\pm e_k}_2 =\prod_{\substack{j=1 \\ j\neq i}}^N\norm{B_j^{(i)}}_2 \\ &= \prod_{\substack{j=1 \\ j\neq i}}^N\sqrt{\sum_{\ell=1}^N\left(B_{j,\ell}^{(i)}\right)^2} 
    \leq\prod_{\substack{j=1 \\ j\neq i}}^N\sqrt{\sum_{\ell=1}^N\norm{B}_{\infty}^2} =\prod_{\substack{j=1 \\ j\neq i}}^N\sqrt{N\norm{B}_{\infty}^2} =\prod_{\substack{j=1 \\ j\neq i}}^N\sqrt{N}\norm{B}_{\infty} \\ &=(\sqrt{N}\norm{B}_{\infty})^{N-1}
\end{align*}
using the Hadamard inequality for determinants~\cite{Hadamard1893}. $B^{(i)}_j$ denotes the $j$-th column of $B^{(i)}$ and $B^{(i)}_{j,\ell}$ is the $\ell$-th entry in that column. Using the fact that $\norm{B}_{\infty}\leq\norm{A}_{\infty}$ and scaling up each $v\in S$ to an integer vector $z$ concludes the proof. Note that we get an extra factor $N$ in the $\ell_1$-norm bound.
\end{proof}

\begin{proof}[\Cref{thm:reducibilitylowerbound}]
For every $i\in[N-1]$, let $v^{(i)}$ be the $N$-dimensional vector with all entries zero except $v^{(i)}_i=\Delta$ and let $u^{(i)}$ be the $N$-dimensional vector with all entries zero except $u^{(i)}_{i+1}=1$. Note that $v^{(i)},u^{(i)}\in[-\Delta,\Delta]^N$. For every $i\in[N-1]$, we have $w^Tv^{(i)}=\Delta^{i-1}\Delta=\Delta^i=w^Tu^{(i)}$. Let $\bar{w}\in\Z^N$ be a vector that is equivalent to $w$ on interval $[-\Delta,\Delta]$. Then we also have $\bar{w}^Tv^{(i)}=\bar{w}^Tu^{(i)}$ and hence $\bar{w}_i \Delta=\bar{w}_{i+1}$ for every $i\in[N-1]$. Iterating this equality constraint through all $i\in[N-1]$ (starting at $i+1=N$), we get
\[\bar{w}_N=\bar{w}_{N-1}\Delta=\hdots=\bar{w}_2\Delta^{N-2}=\bar{w}_1\Delta^{N-1}.\]
Since $w$ is non-zero, $\bar{w}$ must also be non-zero; otherwise they would not be equivalent on $[-\Delta,\Delta]^N$. As $\bar{w}$ is integer, the above equalities imply that $|\bar{w}_1|\geq1$; otherwise all the equal terms would be zero. Now, $\norm{\bar{w}}_1\geq |w_N'|=|\bar{w}_1|\Delta^{N-1}\geq \Delta^{N-1}$, which concludes the proof.
\end{proof}
\begin{proof}[\Cref{lem:sign}]
Consider the points $e_i$ and $\mathbf{0}$ in the interval $[-\Delta,\Delta]$. Then $w_i=we_i\leq w\mathbf{0}=0$ if and only if $\bar{w}_i=\bar{w}e_i\leq \bar{w}\mathbf{0}=0$ by the equivalence of $w$ and $\bar{w}$ on the interval $[-\Delta,\Delta]$. By swapping the vectors, we not only get $w_i\leq0\iff\bar{w}_i\leq0$, but also $w_i\geq0\iff\bar{w}_i\geq0$. So the signs of each entry are identical.
\end{proof}
\begin{proof}[\Cref{thm:reduceconstructive}]
Consider again the cone
\[C:=\bigcap_{\substack{(u,v)\in[-\Delta,\Delta]^N\times[-\Delta,\Delta]^N\\wu\geq wv}}\left\{\left.x\in\R^N\,\right|\,(u-v)x\geq0\right\}\]
from \Cref{thm:reducenonconstructive}. As we already established, $w\in C$ and the points in the relative interior of $C$ are equivalent to $w$ on $[-\Delta,\Delta]^N$. Shifting the right-hand-side to 1 instead of 0 for some of the halfspaces (those where $wu\neq wv$) yields a translated cone which corresponds to the relative interior of $C$. The vector $\bar{w}$ can be computed by solving the ILP corresponding to the relative interior of the cone. This can be done using the algorithm by Reis and Rothvoss~\cite{RR23} in time $(\log(N))^{\Oh{N}}((N+1)(2\Delta+1)^{2N}2\Delta))^{O(1)}$, as the ILP has $N$ variables, at most $((2\Delta+1)^N)^2$ inequality constraints and each coefficient is bounded by $2\Delta$. This running time is in $(N\Delta)^{\Oh{N}}$. It is important to note that the ILP is bounded if we minimize the $\ell_1$-norm of the solution, which can be done with a linear objective, because we know the sign of each entry in the solution (because we know the signs of the entries of $w$ and they have to be equal by \Cref{lem:sign}).

\end{proof}
\begin{proof}[\Cref{thm:ilpkernel2}]
Consider a feasibility ILP of the form
\begin{align*}
    \mathcal{A}x&= b \\
    x&\in\N^N
\end{align*}
with $\mathcal{A}\in\Z^{M\times N}$ and let $\Delta = \norm{\mathcal{A}}_{\infty}$. By the proximity result of Eisenbrand and Weismantel (Theorem \ref{thm:eisweis}), there exists an optimal solution $z^*$ of the ILP such that
\begin{align*}
    \norm{x^*-z^*}_1 \leq M(2M\Delta + 1)^M.
\end{align*}
Let $x^*$ be a solution of the LP-relaxation. We can get $x^*$ with the algorithm of Tardos~\cite{Tardos86} in polynomial time. Because an optimal solution $z^*$ only deviates from $x^*$ by $M(2M\Delta + 1)^M$ in every component, we can pack $\max\{0, \lceil x_i^*-M(2M\Delta+1)^M \rceil\}$ items for each $i\in[N]$ in a preprocessing step. Thus, the following ILP is a kernel of the given ILP.
\begin{align*}
    \mathcal{A}x'&=b' \\
    x'&\in\N^N \\
    0 \leq \ &x'_i \leq 2M(2M\Delta + 1)^M \text{ for all } i\in [N]
\end{align*}
with $b'= b - \mathcal{A}(\max\{\mathbf{0},\lceil x^*-p \rceil\})$ and $p = \begin{pmatrix}
    M(2M\Delta + 1)^M \\ \vdots \\M(2M\Delta + 1)^M
\end{pmatrix}$. So $0 \leq b'_i \leq N\Delta 2 M(2M\Delta + 1)^M$ holds for all $i \in [N]$.

This kernel has size
\begin{align*}
   O(MN\log( N\Delta 2 M(2M\Delta + 1)^M))\ =O(M^2N\log(NM\Delta))
\end{align*}
which concludes the proof.
\end{proof}

\begin{proof}[\Cref{cor:knapsackkernel}]
This follows directly from \Cref{thm:ilpkernel} using $U:=1$, $N:=n+2$ and $M=2$. Note that the decision version of \KS has two inequality constraints -- one for the capacity and one for the profit -- and hence we also need two slack variables to get equality constraints.
\end{proof}

\begin{proof}[\Cref{cor:ssskernel}]
This follows directly from \Cref{cor:knapsackkernel}, as \SSS is a special case of \KS where all items have profit $p_i = w_i,i\in[n]$.
\end{proof}

\begin{proof}[\Cref{cor:ukskernel}]
Consider an instance of \UKS, i.e., profits $p_1,\hdots,p_n$ and threshold $T$, weights $w_1,\hdots,w_n$ and a capacity $C$. The goal is to solve the following ILP:
\begin{align*}
    \sum_{i=1}^np_ix_i\geq T \\
    \sum_{i=1}^nw_ix_i\leq C \\
    x_i\in\N \quad \forall i\in[n]
\end{align*}
Using a classical reduction from \UKS to \KS (see, e.g.,~\cite{CMWW19}), we get the following form
\begin{align*}
    \sum_{i=1}^n\sum_{j=1}^{K_i}p_{i,j}x_{i,j}\geq T\\
    \sum_{i=1}^n\sum_{j=1}^{K_i}w_{i,j}x_{i,j}\leq C \\
    x_{i,j}\in\{0,1\} \quad \forall i\in[n],j\in[K_i]
\end{align*}
where $K_i:=\lfloor\log(\frac{C}{w_i})\rfloor$, $w_{i,j}:=2^jw_i$ and $p_{i,j}:=2^jp_i$. Note that an item $i$ can only be picked $\lfloor\frac{C}{w_i}\rfloor$ times without over-packing. So the idea of the reduction is to create copies of each item such that instead of picking an item $i$ $k$-times, we can pick the copies of $i$ corresponding to $k$ in binary. Since $K_j\leq \log(C)$, the parameters of the constructed instance are as follows:
\begin{itemize}
    \item $n'\leq n\log(C)$
    \item $w_{max}'\leq w_{max}C$
    \item $p_{max}'\leq p_{max}C$
\end{itemize}
Using the static equivalent instances for \KS from \Cref{cor:knapsackkernel}, we get an equivalent instance of size
\begin{align*}
    \Oh{n'^2\log(n')}=\Oh{n^2\log^2(C)\log(n\log(C))}
\end{align*}
for \UKS.
\end{proof}

\begin{proof}[\Cref{cor:mdimkskernel}]
Recall that in \MDKS, we have $n$ items that are $M$-dimensional and an $M$-dimensional knapsack. We are allowed to pack each item only once, so using \Cref{thm:ilpkernel}, the problem \MDKS has a static equivalent instance of size $\Oh{Mn^2\log(n)}$.
\end{proof}

\subsection{Special ILPs}\label{sec:specialilp}
In this section, we show the implications of \Cref{thm:ilpkernel} for \twostage and \nfold.
\begin{definition}
    A \twostage is an ILP where the matrix $\mathcal{A}$ has the following structure $\mathcal{A}_1$. Here, $A_i \in \mathbb{Z}^{s\times r}$ and $B_i \in \mathbb{Z}^{s \times t}$ for all $i \in [n]$. A \nfold is an ILP where the matrix $\mathcal{A}$ has the following structure $\mathcal{A}_2$. Here, $A_i \in \mathbb{Z}^{r \times t}$ and $B_i \in \mathbb{Z}^{s \times t}$ for all $i \in [n]$.
\end{definition}
    \begin{tabular}{rc}
    $\mathcal{A}_1 := \begin{pmatrix}
A_1 & B_1 & \mathbf{0} & \ldots & \mathbf{0} \\
A_2 & \mathbf{0} & B_2 & & \vdots \\
\vdots & \vdots &  & \ddots & \mathbf{0} \\
A_n & \mathbf{0} & \ldots & \mathbf{0} & B_n
\end{pmatrix}$ & $\mathcal{A}_2:= \begin{pmatrix}
A_1 & A_2   & \ldots & A_n \\
B_1 & \mathbf{0} & \ldots  & \mathbf{0} \\
\mathbf{0} & B_2 &    \ddots &  \vdots\\
\vdots & \ddots & \ddots & \mathbf{0} \\
\mathbf{0} & \ldots & \mathbf{0} &   B_n
\end{pmatrix}$ \\
\twostage & \nfold
\end{tabular}
\begin{corollary}\label{cor:2stage}
\twostage has a static equivalent instance of size $\Oh{ns(r+t)^2\log((r+t)U)}$, where $U$ is an upper bound for the $\ell_{\infty}$-norm of the smallest feasible solution. 
\end{corollary}
\begin{proof}
We use \Cref{thm:ilpkernel} for each row-block. This way, we get $n$ static equivalent instances, each of size $\Oh{s(r+t)^2\log((r+t)U)}$. Hence, this yields a static equivalent instance of size $\Oh{ns(r+t)^2\log((r+t)U)}$. A vector $x^*$ is a solution of the original ILP, if and only if its respective parts are solutions of the row-block ILPs. This is the case if and only if the parts are solutions of the transformed row-block ILPs, which is equivalent to $x^*$ being a solution to the whole transformed ILP.
\end{proof}

\begin{corollary}\label{cor:nfold}
\nfold has a static equivalent instance of size $\Oh{n^2t^2rs\log(ntU)}$, where $U$ is an upper bound for the $\ell_{\infty}$-norm of the smallest feasible solution. 
\end{corollary}
\begin{proof}
Again, we use \Cref{thm:ilpkernel}; once for the linking constraints and then once for each of the $n$ blocks on the diagonal. So we get one static equivalent instance of size $\Oh{r(nt)^2\log(ntU)}$ and $n$ static equivalent instances of size $\Oh{st^2\log(tU)}$ each. So in total, we get a static equivalent instance of size 
\[\Oh{r(nt)^2\log(ntU)+nst^2\log(tU)}=\Oh{n^2t^2rs\log(ntU)}.\]
A vector $x^*$ is a solution of the original ILP, if and only if it is a solution of the linking constraints and its respective parts are solutions of the block-ILPs. Again, this is the case if and only if the parts are solutions of the transformed block-ILPs and linking constraints, which is equivalent to $x^*$ being a solution to the whole transformed ILP.
\end{proof}

\section{Preprocessing Algorithm for \LB}
In this section, we show the existence of equivalent instances for \tf{\P}{}{\{\cmax,\cmin,\envy\}}, in particular we proof \Cref{lem:blemma}. For this, we give a kernel for $\PCmax$. Afterwards, we show how this kernels can be transferred to kernels for \tf{\P}{}{\cmin} and \tf{\P}{}{\envy}. 

\begin{lemma}\label{lem:pcmaxk} 
For $\PCmax$, there exists a kernel where the number of jobs of a specific type on a specific machine is bounded by $8d\pmax+4d$. So the load of every machine is bounded by $8d^2\pmax^2+4d^2\pmax$. The kernelization runs in $\Oh{d}$ time.
\end{lemma}
\begin{proof}
Let $p \in \mathbb{N}^d_{>0}$ be the processing time vector, let $n \in \mathbb{N}^d_{>0}$ be the job vector and let $m \in \mathbb{N}_{>0}$ be the number of machines of a $\PCmax$ instance. Let $T$ be the guessed makespan.
Considering the assignment ILP:
    \begin{align*}
        \sum_{i \in [m]} x_j^i &= n_j, \forall j \in [d] \\
        \sum_{j \in [d]} p_j x_j^i &\le T, \forall i \in [m].
    \end{align*}
    Adding slack variables transforms the equations into inequalities.
    Note that all blocks are the same, i.e.\ we have the following $n$-fold ILP
    \begin{align*}
        \mathcal{A} x= \begin{pmatrix}
            A & \dots & A \\
            B & \dots & 0 \\
            \vdots &\ddots & \vdots \\
            0 & \dots & B \\
        \end{pmatrix} x = b
    \end{align*}
    where $A = \begin{pmatrix}
        I_d & \mathbf{O}\\
    \end{pmatrix}$ and 
    $B = (p_1, \dots, p_d,1)$.
    Also the bricks of the right hand side $b_i = T$ are the same for $i \in \{d+1,\dots,d+m\}$. Therefore, we have just one brick type in the lower part of the ILP (matrices and right hand side are equal for all bricks). In the following, we show that if $\mathcal{A} x=b$ is feasible, there exists a feasible $x := (x^1,\dots,x^m)^T$ with $x^i \in \ZZ^{d+1}$ and $\|x^i - x^{i'}\|_\infty \le 2d \ginfB = O(d \ginfB)$ for all $i,i'\in [m]$. We formulate this as general lemma since we use this also for the kernels for \tf{\P}{}{\cmin} and \tf{\P}{}{\envy}. The proof of the following \Cref{lem:balhelp} is a specialization of the proof idea of lemma 10 of~\cite{KKLMO23}.

\begin{lemma}\label{lem:balhelp}
Let $d,m \in \mathbb{N}$, let $A = \begin{pmatrix}
    I_d & \mathbf{0}
\end{pmatrix}$, let $B \in \mathbb{Z}^{s \times (d +1)}$, let $\mathcal{A}$ be the $n$-fold matrix with $m$ blocks $A$ on the horizontal and $m$ blocks $B$ on the diagonal and let $b \in \mathbb{N}^{d + s \cdot m}$ with the same vector for each block, i.e. $(b_{d+1}, \ldots, b_{d + s}) = (b_{d+s+1}, \ldots, b_{d+2s}) = \ldots = (b_{d + (m-1)\cdot s + 1}, \ldots, b_{d + m \cdot s})$. If $\mathcal{A} x=b$ is feasible, there exists a feasible $x := (x^1,\dots,x^m)^T$ with $x^i \in \ZZ^{d+1}$ and $\|x^i - x^{i'}\|_\infty \le 2d \ginfB = O(d \ginfB)$ for all $i,i'\in [m]$.
\end{lemma}
\begin{proof}
Let $d,m \in \mathbb{N}$, let $A = (I_d\ \mathbf{0})$, let $B \in \mathbb{Z}^{s \times (d +1)}$, let $\mathcal{A}$ be the $n$-fold matrix with $m$ blocks $A$ on the horizontal and $m$ blocks $B$ on the diagonal and let $b \in \mathbb{N}^{d + s \cdot m}$ with the same vector for each block, i.e. $(b_{d+1}, \ldots, b_{d + s}) = (b_{d+s+1}, \ldots, b_{d+2s}) = \ldots = (b_{d + (m-1)\cdot s + 1}, \ldots, b_{d + m \cdot s}) = t$ for $t \in \mathbb{N}^s$. Therefore, we have just one brick type in the lower part of the ILP (matrices and right hand side are equal for all bricks).
    In the following, we show that if $\mathcal{A} x=b$ is feasible, there exists a feasible $x := (x^1,\dots,x^m)^T$ with $x^i \in \ZZ^{d+1}$ and $\|x^i - x^{i'}\|_\infty \le 2d \ginfB = O(d \ginfB)$ for all $i,i'\in [m]$. We refer to $x^i$ as configurations or bricks.
    For a feasible solution $x$, let $C(x) := \{x^1,\dots,x^m\}$ be the multiset of chosen configurations.
    Let $R(x) := \max_{c,c' \in C(x)} \|c - c'\|_\infty$ be the largest infinity-norm distance between two configurations. 
    Also define $\zeta(x) \in \mathbb{Q}^{d+1}$ as the fractional center of those configurations, i.e.\ for each coordinate $j \in [d+1]$, set 
    $$\zeta_j(x) := \min_{c \in C(x)} c_j +  \nicefrac{(\max_{c\in C(x)} c_j - \min_{c \in C(x)} c_j)}{2}.$$ In particular, we get that 
    \begin{align} \label{eq: center dist}
        \|c - \zeta(x)\|_\infty \le  R(x)/2    
    \end{align}
    holds for all $c \in C(x)$. 
    We say that coordinate $j \in [d+1]$ of a configuration $c \in C(x)$ is \emph{tight}, if $|c_j - \zeta_j(x)| =  R(x)/2$ holds. Let $S(x) := \sum_{c \in C(x)} \sum_{j\in \{j|c_j \text{ is tight}\}}1$ be the total number of tight coordinates.
    
    Let $\ginfB$ be an upper bound of the infinity norm of elements of the Graver basis of $B$. 
    Set $L := 2d \ginfB$. This is our desired upper bound on the infinity-norm distance between two configurations. 
    Now, we show that there exists a feasible solution that fulfills this bound.
    Towards a proof of contradiction, assume that $x$ is a feasible solution to $\mathcal{A}x=b$ that minimizes $R(x) =: R$ (i.e.\ there does not exist a feasible solution with $\|c - c'\|_\infty < R$ for all $c \in C(x)$) and with minimal number of tight coordinates $S(x)$ (in that order).
    For contradiction, we now assume $R > L$, i.e.\ there exist two configurations $c,c' \in C(x)$ with $\|c - c'\|_\infty = R > L$. Let $i$ and $i'$ be the corresponding machines to those configurations, i.e.\ $x^i = c$ and $x^{i'} = c'$. Since $\|x^i - x^{i'}\|_\infty = R > L \ge 0$, we have that $x^i \not= x^{i'}$ and $i \not= i'$.

    First, we show that we can redistribute the jobs on the machines such that we get $\|x^i - x^{i'}\|_\infty \le L$.
    Since $B x^i = t$, $B x^{i'} = t$ holds, we get $B (x^i - x^{i'}) = t - t = \mathbf{0}$ and therefore, $x^i - x^{i'}$ is an element in the kernel of the diagonal block. 
    Lemma 3.4 in~\cite{Onn10} states that we can decompose this difference into the sum of at most $2d$ different Graver basis elements that are sign-compatible with $x^i-x^{i'}$. Therefore, we have
    \begin{align*}
        x^i - x^{i'} = \sum_{j=1}^{2d} \alpha_j g_j
    \end{align*} with $\alpha_j \in \mathbb{N}$ and $g_j \sqsubseteq x^i - x^{i'}$ for each $j \in [2d]$.
    Note that the used Graver basis elements are bounded by $\|g_j\|_\infty \le \ginfB$.
    Our assumption states $\|x^i - x^{i'}\|_\infty > 2d\ginfB$. Thus, $\alpha_j > 1$ for at least one $j \in [2d]$.
    
    Next, the jobs on the machines can be redistributed. We only consider the Graver basis elements $g_j$ with $\alpha_j >1$, i.e.\ where we can subtract the Graver basis element from one brick and add it to the other one. This corresponds to moving some jobs from one machine to another one. This can be done because after the redistribution, we still assign the same amount of jobs to the same amount of machines. Also, $g_j \sqsubseteq x^i - x^{i'}$ for all $j \in [2d]$ ensures that the jobs are on the machine they are taken from. Therefore, we set
    \begin{equation}\label{eq:redist}
        \begin{array}{ll}
            h &:= \sum_{j=1}^{2d} \lfloor\alpha_j / 2\rfloor g_j \\
            \bax^i &:= x^i - h \\
            \bax^{i'} &:= x^{i'} + h.
        \end{array}
    \end{equation}

    After the redistribution, we have 
    \begin{align*}
        \bax^i - \bax^{i'} &= x^i - x^{i'} -2h \\
        &= \sum_{j=1}^{2d} \alpha_j g_j - 2\sum_{j=1}^{2d} \lfloor\alpha_j / 2\rfloor g_j \\
        &= \sum_{j=1}^{2d} (\alpha_j - 2\lfloor\alpha_j / 2\rfloor) g_j 
    \end{align*}
    Set $\bar{\alpha}_j := \alpha_j - 2\lfloor\alpha_j / 2\rfloor$. Note that $\bax^i - \bax^{i'} = \sum_{j=1}^{2d} \bar{\alpha}_j g_j$ with $\bar{\alpha}_j \le 1$ and $g_j \sqsubseteq \bax^i - \bax^{i'}$.
    This implies $\|\bax^i - \bax^{i'}\|_\infty \le 2d\ginfB = L$.
    Due to the sign compatibility of the graver basis elements and $x^i - x^{i'}$, the vectors $\bax^i$ and $\bax^{i'}$ are non-negative. Also, since $B g_j = \mathbf{0}$ for all $j \in [2d]$ we get $Bh = \mathbf{0}$, which implies that $\bax^i$ and $\bax^{i'}$ are feasible integer solutions to $B x = t$.

    Next, we show that the total number of tight coordinates has been reduced with this redistribution. This then is a contradiction to the assumption that $S(x)$ is minimal.
    Since $\|x^i - x^{i'}\|_\infty = R$ holds, we have that there exists a coordinate $j \in [d+1]$ such that $|x_j^i - x_j^{i'}| = R$ holds. Set $\zeta:=\zeta(x)$. With the triangle inequality, we have 
    \begin{align*}
        R &= |x_j^i - x_j^{i'}| \\
        &= |x_j^i - \zeta_j + \zeta_j - x_j^{i'}| \\
        &\le |x_j^i - \zeta_j| + |\zeta_j - x_j^{i'}|
    \end{align*}
    With \eqref{eq: center dist}, we have that $|x_j^i - \zeta_j| \le R/2$ and $|\zeta_j - x_j^{i'}| \le R/2$, which now implies that both, $x_j^i$ and $x_j^{i'}$ are tight. W.l.o.g.\ assume that $x_j^i \ge x_j^{i'}$. 
    After the redistribution, we have $\|\bax^i - \bax^{i'}\|_\infty \le L < R.$ This implies $|\bax_j^i - \bax_j^{i'}| < R$. Thus, $j$ is no longer a tight variable for both, $\bax_j^i$ and $\bax_j^{i'}$. However, we still have $\bax_j^i \ge \bax_j^{i'}$. 
    More concretely, with \eqref{eq:redist} and $g_j \sqsubseteq \bax^i - \bax^{i'}$ it holds for each coordinate $k$ that 
    \begin{align*}
        x_k^i \ge x_k^{i'} &\implies \bax_k^i \ge \bax_k^{i'},\\
        x_k^i \le x_k^{i'} &\implies \bax_k^i \le \bax_k^{i'} \quad \text{and}\\
        |x_k^i - x_k^{i'}| &\ge |\bax_k^i - \bax_k^{i'}|,
    \end{align*}
    All other machines are not altered by this redistribution which then implies that no new tight coordinates are added. Further, after this redistribution, no  configuration $c''$ with $\norm{c''-c}_{\infty}>R$ or $\norm{c''-c'}_{\infty}>R$ exists. This is because, for any coordinate, they either were on the same side of $\zeta_j,$ i.e.\ $(c''_j-\zeta_j)\cdot(c_j-\zeta_j)\geq 0,$ or on opposite sides. If they were on opposite sides, their distance only got smaller and was bounded by $R$ before. If they were on the same side, their distance was bounded by $R/2$ before and remains bounded as such.

    This contradicts the assumption that $S(x)$ is minimal, which implies $R\le L$, i.e.\ $\|c - c'\|_\infty \le 2d\ginfB$ for all configurations $c,c' \in C(x).$
    \end{proof}
    
    Now, we can use this balancing result to do the following $O(d)$-time preprocessing step to achieve a kernel where at most $4d\ginfB$ jobs of each type have to be scheduled on each machine.

    Let $Q_{\text{frac}}$ be the optimal fractional schedule of the instance, i.e.\ if we allow the assignment of fractional multiplicities of jobs to machines, then all machines complete their processing at time $Q_{\text{frac}}.$
    Now, let $x$ be a feasible solution to $\mathcal{A}x=b$ with $R(x) \le L$ and $\|c - \zeta(x)\|_\infty \le L/2$ for all $c \in C(x)$. Due to \Cref{lem:balhelp}, such a solution exists.
    For each job-type $j \in [d]$, preassign $\lceil Q_{\text{frac}} \rceil - 2d\ginfB$ jobs on each machine. As the number of jobs of type $j$ may not exceed $\lfloor Q_{\text{frac}} \rfloor + 2d\ginfB$ on each machine, it remains to schedule at most $4d\ginfB$ jobs of each type on each machine.
    
    By Eisenbrand et al.~\cite{EHK18}, we know $g_1(B) \le 2\pmax+1$. This implies $\ginfB \le 2\pmax+1$. Therefore, we need to schedule at most $4d(2\pmax +1) = 8d\pmax+4d= O(d\pmax)$ jobs of each type on each machine. 
    \end{proof}

We use the same technique to get kernels for \tf{\P}{}{\cmin} and \tf{\P}{}{\envy}.

\begin{lemma}\label{lem:pcmink}
    For \tf{\P}{}{\cmin}, there exists a kernel where the number of jobs of a specific type on a specific machine is bounded by $8d\pmax+4d$. So the load of every machine is bounded by $8d^2\pmax^2+4d^2\pmax$. The kernelization runs in $\Oh{d}$ time.
\end{lemma}
\begin{proof}
     Let $p \in \mathbb{N}^d_{>0}$ be the processing time vector, let $n \in \mathbb{N}^d_{>0}$ be the job vector and let $m \in \mathbb{N}_{>0}$ be the number of machines of a \tf{\P}{}{\cmin} instance. Let $T$ be the guessed minimum completion time.
Considering the assignment ILP:
    \begin{align*}
        \sum_{i \in [m]} x_j^i &= n_j, \forall j \in [d] \\
        \sum_{j \in [d]} p_j x_j^i &\ge T, \forall i \in [m].
    \end{align*}
    Adding slack variables transforms the equations into inequalities.
    Note that all blocks are the same, i.e.\ we have the following $n$-fold ILP
    \begin{align*}
        \mathcal{A} x= \begin{pmatrix}
            A & \dots & A \\
            B & \dots & 0 \\
            \vdots &\ddots & \vdots \\
            0 & \dots & B \\
        \end{pmatrix} x = b
    \end{align*}
    where $A = \begin{pmatrix}
        I_d & \mathbf{O}\\
    \end{pmatrix}$ and 
    $B = (p_1, \dots, p_d,-1)$.
    Also the bricks of the right hand side $b_i = T$ are the same for $i \in \{d+1,\dots,d+m\}$. Therefore, we have just one brick type in the lower part of the ILP (matrices and right hand side are equal for all bricks). By \Cref{lem:balhelp}, we get that if $\mathcal{A} x=b$ is feasible, there exists a feasible $x := (x^1,\dots,x^m)^T$ with $x^i \in \ZZ^{d+1}$ and $\|x^i - x^{i'}\|_\infty \le 2d \ginfB = O(d \ginfB)$ for all $i,i'\in [m]$. We can use this balancing result to do the following $O(d)$-time preprocessing step to achieve a kernel where at most $4d\ginfB$ jobs of each type have to be scheduled on each machine.

    Let $Q_{\text{frac}}$ be the optimal fractional schedule of the instance, i.e.\ if we allow the assignment of fractional multiplicities of jobs to machines, then all machines complete their processing at time $Q_{\text{frac}}.$
    Now, let $x$ be a feasible solution to $\mathcal{A}x=b$ with $R(x) \le L$ and $\|c - \zeta(x)\|_\infty \le L/2$ for all $c \in C(x)$. Due to \Cref{lem:balhelp}, such a solution exists.
    For each job-type $j \in [d]$, preassign $\lceil Q_{\text{frac}} \rceil - 2d\ginfB$ jobs on each machine. As the number of jobs of type $j$ may not exceed $\lfloor Q_{\text{frac}} \rfloor + 2d\ginfB$ on each machine, it remains to schedule at most $4d\ginfB$ jobs of each type on each machine.
    
    By Eisenbrand et al.~\cite{EHK18}, we know $g_1(B) \le 2\pmax+1$. This implies $\ginfB \le 2\pmax+1$. Therefore, we need to schedule at most $4d(2\pmax +1) = 8d\pmax+4d= O(d\pmax)$ jobs of each type on each machine. 
\end{proof}

\begin{lemma}\label{lem:pcenv}
    For \tf{\P}{}{\envy}, there exists a kernel where the number of jobs of a specific type on a specific machine is bounded by $4d(4\pmax +1)^2$. So the load of every machine is bounded by $4d^2\pmax(4\pmax +1)^2$. The kernelization runs in $\Oh{d}$ time.
\end{lemma}
\begin{proof}
     Let $p \in \mathbb{N}^d_{>0}$ be the processing time vector, let $n \in \mathbb{N}^d_{>0}$ be the job vector and let $m \in \mathbb{N}_{>0}$ be the number of machines of a \tf{\P}{}{\envy} instance. Let $T_{\ell}$ be the guessed lower bound on the completion time and let $T_u$ be the guessed upper bound on the completion time.
Considering the assignment ILP:
    \begin{align*}
        \sum_{i \in [m]} x_j^i &= n_j, \forall j \in [d] \\
        \sum_{j \in [d]} p_j x_j^i &\le T_u, \forall i \in [m] \\
        \sum_{j \in [d]} p_j x_j^i &\ge T_{\ell}, \forall i \in [m].
    \end{align*}
    Adding slack variables transforms the equations into inequalities.
    Note that all blocks are the same, i.e.\ we have the following $n$-fold ILP
    \begin{align*}
        \mathcal{A} x= \begin{pmatrix}
            A & \dots & A \\
            B & \dots & 0 \\
            \vdots &\ddots & \vdots \\
            0 & \dots & B \\
        \end{pmatrix} x = b
    \end{align*}
    where $A = \begin{pmatrix}
        I_d & \mathbf{O}\\
    \end{pmatrix}$ and 
    $B = \begin{pmatrix}
        p_1 & \ldots & p_d & 1 \\
        p_1 & \ldots & p_d & -1
    \end{pmatrix}$.
    Also the bricks of the right hand side $b^i = \begin{pmatrix}
        T_u \\ T_{\ell}
    \end{pmatrix}$ are the same for $i \in \{d+1,\dots,d+m\}$. Therefore, we have just one brick type in the lower part of the ILP (matrices and right hand side are equal for all bricks). By \Cref{lem:balhelp}, we get that if $\mathcal{A} x=b$ is feasible, there exists a feasible $x := (x^1,\dots,x^m)^T$ with $x^i \in \ZZ^{d+1}$ and $\|x^i - x^{i'}\|_\infty \le 2d \ginfB = O(d \ginfB)$ for all $i,i'\in [m]$. We can use this balancing result to do the following $O(d)$-time preprocessing step to achieve a kernel where at most $4d\ginfB$ jobs of each type have to be scheduled on each machine.

    Let $Q_{\text{frac}}$ be the optimal fractional schedule of the instance, i.e.\ if we allow the assignment of fractional multiplicities of jobs to machines, then all machines complete their processing at time $Q_{\text{frac}}.$
    Now, let $x$ be a feasible solution to $\mathcal{A}x=b$ with $R(x) \le L$ and $\|c - \zeta(x)\|_\infty \le L/2$ for all $c \in C(x)$. Due to \Cref{lem:balhelp}, such a solution exists.
    For each job-type $j \in [d]$, preassign $\lceil Q_{\text{frac}} \rceil - 2d\ginfB$ jobs on each machine. As the number of jobs of type $j$ may not exceed $\lfloor Q_{\text{frac}} \rfloor + 2d\ginfB$ on each machine, it remains to schedule at most $4d\ginfB$ jobs of each type on each machine.
    
    By Eisenbrand et al.~\cite{EHK18}, we know $g_1(B) \le (4\pmax+1)^2$. This implies $\ginfB \le (4\pmax+1)^2$. Therefore, we need to schedule at most $4d(4\pmax +1)^2 = O(d\pmax^2)$ jobs of each type on each machine. 
\end{proof}

\balance*
\begin{proof}
Since $8d\pmax + 4d \leq 4d(4\pmax +1)^2$, this follows directly by \Cref{lem:pcmaxk}, \Cref{lem:pcmink} and \Cref{lem:pcenv}.
\end{proof}

\end{document}

%% file: preambleIwoca.tex
\makeatletter
\newcommand{\thickhline}{%
    \noalign {\ifnum 0=`}\fi \hrule height 1.5pt
    \futurelet \reserved@a \@xhline
}
\newcolumntype{"}{@{\hskip\tabcolsep\vrule width 1pt\hskip\tabcolsep}}
\makeatother

\newcommand{\Z}{\mathbb{Z}}
\newcommand{\R}{\mathbb{R}}
\newcommand{\Oh}[1]{O(#1)}
\newcommand{\Ohtilde}[1]{\tilde{O}(#1)}
\newcommand{\Ah}{\mathcal{A}}
\newcommand{\poly}[1]{\text{poly}(#1)}
\newcommand{\etal}{et~al.\xspace}
\newcommand{\oh}[1]{o\left(#1\right)}
\newcommand{\Omg}[1]{\Omega\left(#1\right)}
\newcommand{\norm}[1]
{\left\lVert#1\right\rVert}
\newcommand{\eps}{\varepsilon}
\newcommand{\I}{\mathcal{I}}
\newcommand{\closedinterval}[1]{\left[#1\right]}

\newcommand*{\shortautoref}[1]{%
  \begingroup
    \def\sectionautorefname{Sec.}%
    \def\theoremautorefname{Thm.}%
    \def\corollaryautorefname{Cor.}%
    \def\lemmaautorefname{Lem.}%
    \def\propositionautorefname{Prop.}%
    \autoref{#1}%
  \endgroup
}

\def\pmax{p_{\max}} 
\def\smax{s_{\max}} 
\def\dmax{d_{\max}} 
\def\wmax{w_{\max}} 
\def\nmax{n_{\max}} 
\def\m{m} 
\def\n{\norm{n}_1} 
\def\opt{\mathrm{OPT}}

\renewcommand{\matrix}[1]{\begin{pmatrix}#1\end{pmatrix}}

\newcommand{\rowsA}{s}
\newcommand{\colsA}{r}
\newcommand{\rowsB}{\rowsA}
\newcommand{\colsB}{t}

\newcommand{\KS}{\textsc{Knapsack}\xspace}
\newcommand{\UKS}{\textsc{Unbounded-Knapsack}\xspace}
\newcommand{\SSS}{\textsc{SubsetSum}\xspace}
\newcommand{\MDKS}{\textsc{Multidimensional-Knapsack}\xspace}
\newcommand{\twostage}{\textsc{2-Stage Stochastic ILP}\xspace}
\newcommand{\nfold}{\textsc{$n$-Fold ILP}\xspace}
\newcommand{\lp}{\textsc{Feasibility LP}\xspace}
\newcommand{\N}{\mathbb{N}}
\newcommand{\Nwithzero}{\mathbb{N}_0}
\newcommand{\Nwithoutzero}{\mathbb{N}_{>0}}
\newcommand{\mswbp}{\textsc{MinSum-Weighted-BinPacking}\xspace}
\newcommand{\PART}{\textsc{Partition}\xspace}

\newcommand{\zero}{\textbf{0}}

\newcommand{\lr}[1]{\multicolumn{1}{|c}{#1}}
\newcommand{\rr}[1]{\multicolumn{1}{c|}{#1}}
\newcommand{\lrr}[1]{\multicolumn{1}{|c|}{#1}}

\newcommand{\one}{1}
\newcommand{\two}{P2}
\newcommand{\three}{P3}
\newcommand{\four}{P4}
\newcommand{\fixed}{Pm}
\renewcommand{\P}{P}
\newcommand{\Q}{Q}
\newcommand{\unrel}{R}

\newcommand{\size}{size}
\newcommand{\any}{any}
\newcommand{\rej}{Rej\leq Q}
\newcommand{\rj}{r_{j}}
\newcommand{\jobtypes}{d} 
\renewcommand{\dj}{d_{j}}
\newcommand{\bounded}{x_{i,j}\leq c}
\newcommand{\class}{\textup{class}}
\newcommand{\setup}{s_{j}}
\newcommand{\capacity}{\textup{cap}}
\renewcommand{\vec}{\textup{vec}}

\newcommand{\cmax}{C_{\max}}
\newcommand{\cmin}{C_{\min}}
\newcommand{\envy}{C_{\textup{envy}}}
\newcommand{\lmax}{L_{\max}}
\newcommand{\tmax}{T_{\max}}
\newcommand{\fmax}{F_{\max}}
\newcommand{\sumc}{\sum C_{j}}
\newcommand{\sumu}{\sum U_{j}}
\newcommand{\sumv}{\sum V_{j}}
\newcommand{\sumt}{\sum T_{j}}
\newcommand{\sumf}{\sum F_{j}}
\newcommand{\sumwu}{\sum w_{j}U_{j}}
\newcommand{\sumwv}{\sum w_{j}V_{j}}
\newcommand{\sumwt}{\sum w_{j}T_{j}}
\newcommand{\sumwc}{\sum w_{j}C_{j}}
\newcommand{\sumwf}{\sum w_{j}F_{j}}
\newcommand{\elltwo}{\ell_2}

\newcommand{\tf}[3]{\({#1|#2|#3}\)}

\newcommand{\assip}{\textsc{Assignment ILP}\xspace}
\newcommand{\confip}{\textsc{confILP}\xspace}
\newcommand{\confips}{\textsc{confILP}s\xspace}

\def\mlm{m} 
\def\nlm{\n} 
\def\M{\closedinterval{\mlm}} 
\def\J{\closedinterval{\nlm}} 
\def\jk{\mathcal{J}_k} 
\def\D{\closedinterval{\jobtypes}} 
\newcommand{\slotmilpdelta}{\operatorname{slot-MILP^{\delta}}}
\newcommand{\slotmilp}{\operatorname{slot-MILP}}
\newcommand{\enc}[1]{\langle#1\rangle}
\newcommand{\epsfrac}{\frac{1}{\eps}}
\newcommand{\fromvertex}{v_1}
\newcommand{\tovertex}{v_2}
\newcommand{\thirdvertex}{v_3}

\newcommand{\preprocessingrt}{d}
\newcommand{\dprt}{\left( m \cdot \jobtypes \cdot \pmax \right)^{\O(\epsfrac)}}
\newcommand{\swaprt}{\epsfrac \cdot m^2 + \log(\jobtypes \cdot \pmax)}
\newcommand{\localsearchrt}{m^3 \cdot \jobtypes^3 \cdot \pmax \cdot \epsfrac \cdot \log(m \cdot \jobtypes \cdot \pmax)}
\newcommand{\longadditiveapproxrt}{\Oh{\dprt}}
\newcommand{\additiveapproxrt}{ (m \cdot \pmax) ^{\Oh{\epsfrac}}}
\newcommand{\originaladditiveapproxrt}{ m \cdot \n^{\Oh{\epsfrac}}}

\renewcommand{\det}[1]{\text{det}\left(#1\right)}

\newcommand{\overbar}[1]{\mkern 1.5mu\overline{\mkern-1.5mu#1\mkern-1.5mu}\mkern 1.5mu}

\newcommand{\OPT}{\operatorname{\text{\textsc{opt}}}}

\newenvironment{myproblem}%
{%
  \leavevmode\nobreak\par
\begin{samepage}
	\begin{list}%
		{}%
		{%
			\def\labelstyle{\itshape}
			\setlength{\topsep}{0pt}%
			\renewcommand\makelabel[1]{%
				\mbox{\normalfont ##1}\hfil
			}%
			\settowidth{\labelwidth}{\labelstyle Parameter:}%
			\setlength{\leftmargin}{\labelwidth}%
			\addtolength{\leftmargin}{\labelsep}%
			\setlength{\itemsep}{0pt}%
			\setlength{\parsep}{0pt}%
		}%
		\def\instance{\item[\labelstyle Instance:]}%
    \def\question{\item[\labelstyle Question:]}%
    \def\task{\item[\labelstyle Task:]}%
		\def\result{\item[\labelstyle Result:]}%
	}%
	{%
	\end{list}%
\end{samepage}
}

\newcommand{\mpqr}{M-PQ-R\xspace}
\newcommand{\mpqmr}{M-PQm-R\xspace}
\newcommand{\mpqrs}{M-PQ-Rs\xspace}
\newcommand{\mpqmrs}{M-PQm-Rs\xspace}
\newcommand{\mpqmmrs}{M-PQ(m)-Rs\xspace}
\newcommand{\mpqmmr}{M-PQ(m)-R\xspace}

\newcommand{\pqr}{PQ-R\xspace}
\newcommand{\pqmr}{PQm-R\xspace}
\newcommand{\pqrs}{PQ-Rs\xspace}
\newcommand{\pqmrs}{PQm-Rs\xspace}
\newcommand{\pqmmrs}{PQ(m)-Rs\xspace}
\newcommand{\pqmmr}{PQ(m)-R\xspace}

\newcommand{\CS}{\textsc{CuttingStock}\xspace}
\newcommand{\BP}{\textsc{BinPacking}\xspace}

\newcommand{\powIteri}{{(i)}}
\newcommand{\powIterell}{{(\ell)}}
\newcommand{\powIterZero}{{(0)}}
\newcommand{\powIterOne}{{(1)}}
\newcommand{\powIterTwo}{{(2)}}
\newcommand{\powIterCalI}{{(\calI)}}
\newcommand{\powIteriminOne}{{(i-1)}}
\newcommand{\powIteriplusOne}{{(i+1)}}
\newcommand{\calI}{\mathcal{I}}
\newcommand{\calIFormula}{
        \begin{cases}
            \log \big(\frac{m_{\max} + \maxSupp}{2\maxSupp+1} \big) +2, &\txtif \log \big(\frac{m_{\max} + \maxSupp}{2\maxSupp+1} \big) \text{ is integer}\\
            \big\lceil\log \big(\frac{m_{\max} + \maxSupp}{2\maxSupp+1} \big)\big\rceil +1, &\txtotherwise.
        \end{cases}}
\newcommand{\calIkFormula}{
        \begin{cases}
            \log \big(\frac{m_k + \maxSupp}{2\maxSupp+1} \big) +2, &\txtif \log \big(\frac{m_k + \maxSupp}{2\maxSupp+1} \big) \text{ is integer}\\
            \big\lceil\log \big(\frac{m_k + \maxSupp}{2\maxSupp+1} \big)\big\rceil +1, &\txtotherwise.
        \end{cases}}
\newcommand{\calO}{O}
\newcommand{\calL}{\mathcal{L}}
\newcommand{\calA}{\mathcal{A}}
\newcommand{\calC}{\mathcal{C}}
\newcommand{\DeltaBound}{\pmax^{\calO(1)}}
\newcommand{\pmaxD}{\pmax^\Od}
\newcommand{\pmaxOne}{\pmax^\Oone}
\newcommand{\Od}{{\calO(d)}}
\newcommand{\Oone}{{\calO(1)}}
\newcommand{\trueCapital}{\texttt{true}}
\newcommand{\falseCapital}{\texttt{false}}
\newcommand{\nullCapital}{\texttt{null}}
\newcommand{\lcm}{\textup{lcm}}
\newcommand{\polyI}{|\calL|^{O(1)}}
\newcommand{\Cpp}{C\nolinebreak\hspace{-.05em}\raisebox{.4ex}{\tiny\bf +}\nolinebreak\hspace{-.10em}\raisebox{.4ex}{\tiny\bf +}}

\newcommand{\supp}{\textup{supp}}
\newcommand{\maxSupp}{K}
\newcommand{\maxSuppFormula}{\lfloor 2 \cdot (d + 1)\cdot \log(4 \cdot (d + 1) \cdot \Delta) \rfloor}

\newcommand{\txtotherwise}{\textup{otherwise}}

\newcommand{\Cmax}{C_{\max}}
\newcommand{\Cmin}{C_{\min}}
\newcommand{\Cenvy}{C_{\textup{envy}}}

\newcommand{\PCmax}{P\|\Cmax}
\newcommand{\PCmin}{P\|\Cmin}
\newcommand{\PCenvy}{P\|\Cenvy}
\newcommand{\PCmaxminenvy}{P\|\{\Cmax,\Cmin,\Cenvy\}}
\newcommand{\QCmax}{Q\|\Cmax}
\newcommand{\QCmin}{Q\|\Cmin}
\newcommand{\QCenvy}{Q\|\Cenvy}
\newcommand{\QCmaxminenvy}{Q\|\{\Cmax,\Cmin,\Cenvy\}}
\newcommand{\RCmax}{R\|\Cmax}
\newcommand{\RCmin}{R\|\Cmin}
\newcommand{\RCenvy}{R\|\Cenvy}
\newcommand{\RCmaxminenvy}{R\|\{\Cmax,\Cmin,\Cenvy\}}

\newcommand{\tiln}{\tilde{n}}
\newcommand{\tilm}{\tilde{m}}
\newcommand{\tilx}{\tilde{x}}
\newcommand{\tilcalL}{\tilde{\calL}}
\newcommand{\hatn}{\hat{n}}
\newcommand{\hatm}{\hat{m}}
\newcommand{\hatx}{\hat{x}}
\newcommand{\hatcalL}{\hat{\calL}}

\newcommand{\RTdp}{\text{RT}_\text{DP}}
\newcommand{\RTpcmax}{\text{RT}_{\PCmax}}
\newcommand{\RTmainalgo}{\text{RT}_\text{feasibility}}

\newcommand{\RTpcmaxO}{\pmaxD \cdot \log(\|n\|_1) + \calO(d)}\newcommand{\RTdpO}{\calO(\tau \cdot (\nu + 1)^d \cdot \RTpcmax) +\tau^2 \cdot (\nu +1)^{\calO(d)} \cdot \log(\tau \cdot (\nu + 1)^d)}

\newcommand{\Pf}{P\|f}
\newcommand{\RTidentical}{\text{RT}_{\Pf}}

\newcommand{\RTopt}{\text{RT}_{\opt}}

\newcommand{\optLf}{\opt_{\calL,f}}
\newcommand{\optn}{\opt_{n}}
\newcommand{\optnOne}{\opt_{n^\powIterOne}}
\newcommand{\optnTwo}{\opt_{n^\powIterTwo}}
\newcommand{\optni}{\opt_{n^\powIteri}}
\newcommand{\optniminOne}{\opt_{n^\powIteriminOne}}
\newcommand{\optnCali}{\opt_{n^\powIterCalI}}
\newcommand{\optq}{\opt_{q}}

\newcommand{\optLfOne}{\opt_{\calL^\powIterOne,f}}
\newcommand{\optLfi}{\opt_{\calL^\powIteri,f}}
\newcommand{\optLfiminone}{\opt_{\calL^\powIteriminOne,f}}
\newcommand{\optLfhat}{\opt_{\hat{\calL}^\powIteri,f}}
\newcommand{\optLftil}{\opt_{\tilde{\calL}^\powIteri,f}}
\newcommand{\optPff}{\opt_{\Pf,f}}

\newcommand{\txtif}{\textup{if }}
\newcommand{\txtother}{\textup{otherwise }}

\newcommand{\smin}{s_{\min}}
\newcommand{\pmin}{p_{\min}}
\newcommand{\nmin}{n_{\min}}
\newcommand{\mmax}{m_{\max}}
\newcommand{\mmin}{m_{\min}}

\newcommand{\Cavg}{C_{\textup{avg}}}

\newcommand{\LB}{\textsc{LoadBalancing}\xspace}

\newcommand{\nullvec}{\mathbf{O}}
\newcommand{\onevec}{\mathbf{1}}

\newcommand{\ZZ}{\mathbb{Z}}
\newcommand{\ZZgeqzero}{\N}
\newcommand{\ZZgeqone}{\mathbb{Z}_{\geq 0}}

\makeatletter
\AfterEndEnvironment{algorithm}{\let\@algcomment\relax}
\AtEndEnvironment{algorithm}{\kern2pt\hrule\relax\vskip3pt\@algcomment}
\let\@algcomment\relax
\newcommand\algcomment[1]{\def\@algcomment{\footnotesize#1}}

\renewcommand\fs@ruled{\def\@fs@cfont{\bfseries}\let\@fs@capt\floatc@ruled
  \def\@fs@pre{\hrule height.8pt depth0pt \kern2pt}%
  \def\@fs@post{}%
  \def\@fs@mid{\kern2pt\hrule\kern2pt}%
  \let\@fs@iftopcapt\iftrue}
\makeatother

\tikzset{%
	>={Latex[width=2mm,length=2mm]},
	base/.style = {rectangle, rounded corners, draw=black,
		minimum width=4cm, minimum height=1cm,
		text centered, font=\sffamily},
	activityStarts/.style = {base, fill=blue!30},
	startstop/.style = {base, fill=red!30},
	activityRuns/.style = {base, fill=subbi!30},
	process/.style = {base, minimum width=2.5cm, fill=orange!15,
		font=\ttfamily},
}
\tikzset{toprule/.style={%
		execute at end cell={%
			\draw [line cap=rect,#1] (\tikzmatrixname-\the\pgfmatrixcurrentrow-\the\pgfmatrixcurrentcolumn.north west) -- (\tikzmatrixname-\the\pgfmatrixcurrentrow-\the\pgfmatrixcurrentcolumn.north east);%
		}
	},
	bottomrule/.style={%
		execute at end cell={%
			\draw [line cap=rect,#1] (\tikzmatrixname-\the\pgfmatrixcurrentrow-\the\pgfmatrixcurrentcolumn.south west) -- (\tikzmatrixname-\the\pgfmatrixcurrentrow-\the\pgfmatrixcurrentcolumn.south east);%
		}
	}
}
\newcommand{\algorithmautorefname}{Algorithm}

\newfloat{algorithm}{t}{lop}
\floatname{algorithm}{Algorithm}

\tikzset{
    cross/.pic = {
    \draw[rotate = 45] (-#1,0) -- (#1,0);
    \draw[rotate = 45] (0,-#1) -- (0, #1);
    }
}